\documentclass[12pt,a4paper]{article}
\usepackage{bbm}
\usepackage{stmaryrd}
\usepackage{amsmath}
\usepackage{amssymb}
\usepackage{hyperref}
\usepackage[mathcal]{euscript}

\newcounter{thMM}[section]
\newcounter{leMM}[section]
\newcounter{deFF}[section]
\newcounter{exMP}[section]
\newcounter{prOP}[section]
\newcounter{coRR}[section]
\newcounter{coexMP}[section]
\newenvironment{theorem}[1][Theorem]{\refstepcounter{thMM}\trivlist
   \item[{\bf #1~\thesection.\arabic{thMM}.}]\it\hskip3pt}{\endtrivlist}

\newenvironment{definition}[1][Definition]{\refstepcounter{deFF}\trivlist
   \item[{\bf #1~\thesection.\arabic{deFF}.}]\rm\hskip3pt}{\endtrivlist}

\newenvironment{proposition}[1][Proposition]{\refstepcounter{prOP}\trivlist
   \item[{\bf #1~\thesection.\arabic{prOP}.}]\it\hskip3pt}{\endtrivlist}
\newenvironment{corollary}[1][Corollary]{\refstepcounter{coRR}\trivlist
   \item[{\bf #1~\thesection.\arabic{coRR}.}]\it\hskip3pt}{\endtrivlist}

\newenvironment{proof}[1][Proof]{\begin{trivlist}
\item[\hskip \labelsep {\bfseries #1}] }{ \begin{flushright}$\square$\end{flushright}\end{trivlist}}
\newenvironment{remark}[1][Remark]{\begin{trivlist}
\item[\hskip \labelsep {\bfseries #1}]}{\end{trivlist}}

\newcommand{\D}{\mathcal{D}}

\DeclareMathOperator{\Vect}{Vect}

\DeclareMathOperator{\Span}{Span}

\begin{document}
\author{Andrew James Bruce \\  \small{Pembrokeshire College,}\\
\small{Haverfordwest,}\\\small{ Pembrokeshire, SA61 1SZ, UK}
\\ \small{\emph{email:} \texttt{andrewjamesbruce@googlemail.com}  } }
\date{\today}
\title{Higher contact-like structures and supersymmetry}
\maketitle

\begin{abstract}
We establish a relation between   higher contact-like structures on supermanifolds and the  $\mathcal{N}=1$  super-Poincar\'e group via its superspace realisation. To do this we introduce a vector-valued contact structure, which we  refer to as a polycontact structure.
\end{abstract}

\section{Introduction}\label{sec:intro}

In this paper we link $\mathcal{N}=1$ supersymmetry, as formulated in superspace, with  a higher order version of contact geometry. We will show that one needs to pass from one-forms to vector-valued  one-forms in order to construct a contact-like form on  $(4|4)$ dimensional flat or rigid  superspace. In analogy with  G$\ddot{\textnormal{u}}$nther's polysymplectic geometry  \cite{Gunther1987}  we will call  ``vector-valued contact forms"  \emph{polycontact forms}.  The link between contact structures supersymmetry originates with the work of Manin \cite{Manin1991} and  was explored by Schwarz \cite{Schwarz1992} and his collaborators in relation with superconformal field theories. We revisited this relation between contact structures on supermanifolds and the    $d=1$, $N=2$ super-Poincar\'e algebra in \cite{Bruce2011}. In particular we  emphasised  how the contact structure arises from more standard coset space methods.   \\

With the relation between SUSY and contact structures being our primary goal here, let us present a lightning review of classical contact structures on manifolds  highlighting the elements we need later. Recall that a precontact structure on a manifold is a one-form that is nowhere vanishing\footnote{We assume all structures to be global and skip questions of orientability.}. Associated with every precontact structure on a manifold is a hyperplane distribution, that is  a subbundle of the tangent bundle of corank $1$. The hyperplane distribution is defined to be  the kernel of the precontact structure.  That is if we denote the precontact structure as $\alpha \in \Omega^{1}(M)$ then $\D_{\alpha} = \ker{\alpha}$. That is the hyperplane distribution consists of all vector fields $X \in \Vect(M)$ such that $i_{X}\alpha =0$.\\

A contact structure on a manifold of dimension $(2n +1)$ ($n \in \mathbbmss{N}^{*})$ is a precontact structure with the extra requirement that the  exterior derivative is a two-form that is  non-degenerate on the hyperplane distribution defined by the structure. That is there are no non-zero vector fields $X \in \D_{\alpha}$ such that $i_{X}(d \alpha)=0$.\\

Let $(M, \alpha)$ be a contact manifold. A diffeomorphism $\phi: M \rightarrow M$  is said to be a \emph{contactomprohism} if and only if $\phi^{*}\alpha = f \alpha$ for some nowhere vanishing function $f \in C^{\infty}(M)$. A diffeomorphism $\phi: M \rightarrow M$  is said to be a \emph{strict contactomorphism} if and only if $\phi^{*}\alpha = \alpha$.  Contactomorphisms preserve hyperplane distributions.  A vector field $X \in \Vect(M)$ is said to be a \emph{contact vector field} if and only of $L_{X}\alpha = f \alpha$. If $f = 0$, then then vector field is said to be a \emph{strict contact vector field}.\\

We do not employ anything deep from the general theory of contact structures and so direct the reader that is unfamiliar with contact structures to introductory texts. For example see Appendix 4 of Arnold's book \cite{Arnold1989} or Etnyre's lectures \cite{Etnyre2003}. Grabowski \cite{Grabowski2011} provides a  good description of contact structures on supermanifolds and their idiocrasies. Also note the  work of Mehta on graded contact structures \cite{Mehta2011}.  \\

 G$\ddot{\textnormal{u}}$nther \cite{Gunther1987} initiated the  study of polysymplectic structures, understood as a natural generlisation of symplectic geometry in which symplectic forms are replaced with vector-valued two forms. The motivation was to understand a Hamiltonian formulation of classical field theories.  In  G$\ddot{\textnormal{u}}$nther's initial work the differential forms are vector field valued, that is the vector bundle in question is the tangent bundle, though this can be generalised to arbitrary vector bundles. It may be interesting to consider  polysymplectic structures on vector bundles with a little extra structure such as Lie algebroids or skew algebroids.\\

In this work we develop a vector-valued version of contact geometry in order to satisfy our needs: \emph{that is to formulate a novel geometric  view  of $\mathcal{N}=1$ supersymmetry}. It is expected that  further refinements  will be required to understand the general theory of polycontact geometry on supermanifolds.  The subtleties and idiosyncrasies of working on supermanifolds, or indeed graded supermanifolds would require careful handling in developing the notion of polycontact structures. We make no attempt to be anything like general in this work. For example, we make no claim on the existence of a Darboux-like theorem in general. Indeed, for multisymplectic and polysymeplectic structures the existence of a Darboux-like coordinates is not guaranteed in general. It is rather an open question to mathematically formulate precisely a higher order version of symplectic geometry suitable for physical applications.  \\

We will make heavy use of local coordinates in order to keep this work very explicit and accessible. Informally one should understand supermanifolds as ``manifolds" with both commuting and anticommuting coordinates. By using local coordinates one avoids the language of sheaf theory and locally ringed spaces.   For further details  about supermanifolds  we recommend \cite{Manin1997,Varadarajan2004,Voronov1992}. \\

This paper is arranged as follows. In \S(\ref{sec:N1}) we present the basics of $\mathcal{N}=1$ supersymmetry as required throughout this paper. The main results of this paper are in \S(\ref{sec:polycontact}) where we define and explore the notion of polycontact geometry applied to supersymmetry. The procedure of ``polysymplectization" is given in \S(\ref{sec:polysymplectic}). We end this paper with a few concluding remarks in \S(\ref{sec:conclusion}). A brief appendix on the basics of vector-valued differential forms on supermanifolds is included.

\begin{remark}
While this paper was being compiled the author was made aware of the work of E. van Erp \cite{Erp2011} who defines the notion of a polycontact structure on a classical manifold in way very similar to the notion employed here.  Importantly van Erp shows how his notion of a polycontact structure is equivalent to the existence of a generalised Szeg\"{o} projection in the Heisenberg algebra  of pseudodifferential operators. It would be very interesting to generalise this result to supermanifolds.
\end{remark}

\section{$\mathcal{N}=1$ supersymmetry in superspace}\label{sec:N1}

We will assume that the reader is familiar with the basic ideas of supersymmetric field theory and in particular the notion of superspace. There are many excellent reviews of the subject, see for example \cite{Wess1992,West1990}.  Consider the superspace $\mathbbmss{R}^{4|4}$ equipped with local coordinates $(x^{\mu}, \theta^{a}, \bar{\theta}_{\dot{a}})$. Here the coordinate $(x^{\mu})$ are Grassmann even (commuting) as where the coordinates $(\theta^{a}, \bar{\theta}_{\dot{a}})$ are Grassmann odd (anticommuting). The SUSY transformations are \emph{defined} in this superspace to be

\begin{eqnarray}\label{n2susytransform}
x^{\mu} &\rightarrow& x'^{\mu} = x^{\mu}+ i \left(\epsilon^{a}(\sigma^{\mu})_{a}^{\:\: \dot{b}} \bar{\theta}_{\dot{b}} - \theta^{a}(\sigma^{\mu})_{a}^{\:\: \dot{b}} \bar{\epsilon}_{\dot{b}}  \right),\\
\nonumber \theta^{a} &\rightarrow& \theta'^{a} = \theta^{a}+ \epsilon^{a},\\
\nonumber \bar{\theta}_{\dot{a}} &\rightarrow& \bar{\theta}'_{\dot{a}} = \bar{\theta}_{\dot{a}}+ \bar{\epsilon}_{\dot{a}},
\end{eqnarray}

where $\epsilon^{a}$ and $\bar{\epsilon}_{\dot{a}}$ are real Grassmann odd parameters. Here $\sigma^{\mu} := (\mathbbmss{1}, \sigma^{i})$, where $\sigma^{i}$ are the Pauli spin  matrices. The factor of $i = \sqrt{-1}$ is included to ensure that the product of two real Grassmann odd objects is real. Thus, the real nature of space-time is maintained.\\

Thinking of the SUSY transformations as a change in local coordinates these naturally induce transformations of the differentials, as well as the partial derivatives. Let us equip $\Pi T(\mathbbmss{R}^{4|4})$ with natural fibre coordinates $(dx^{\mu}, d\theta^{a}, d\bar{\theta}_{\dot{a}})$. Here $dx$ is Grassmann odd, where $d\theta$ and $d\bar{\theta}$ are Grassmann even. Explicitly the SUSY transformations induce vector bundle automorphisms as

\begin{eqnarray}
dx'^{\mu} &=& dx^{\mu} - i \left(d\theta^{a}(\sigma^{\mu})_{a}^{\:\: \dot{b}}\bar{\epsilon}_{\dot{b}} + \epsilon^{a}(\sigma^{\mu})_{a}^{\:\: \dot{b}}d\bar{\theta}_{\dot{b}}   \right),\\
\nonumber  d\theta'^{a} &=& d \theta^{a},\\
\nonumber d \bar{\theta}'_{\dot{a}} &=& d \bar{\theta}_{\dot{a}}.
\end{eqnarray}

The partial derivatives (or one could think about fibre coordinates on the tangent bundle) transform as

\begin{eqnarray}
\frac{\partial}{\partial x'^{\mu}} &=& \frac{\partial}{\partial x^{\mu}},\\
\nonumber \frac{\partial}{\partial \theta'^{a}} & = & \frac{\partial}{\partial \theta^{a}} + i (\sigma^{\mu})_{a}^{\:\: \dot{b}}\bar{\epsilon}_{\dot{b}} \frac{\partial}{\partial x^{\mu}},\\
\nonumber \frac{\partial}{\partial \bar{\theta}'_{\dot{a}}} &=& \frac{\partial}{\partial \bar{\theta}_{\dot{a}}} + i \epsilon^{b}(\sigma^{\mu})_{b}^{\:\: \dot{a}}\frac{\partial}{\partial x^{\mu}}.
\end{eqnarray}

Let us  as standard introduce the two vector fields

\begin{equation}
Q_{a} =  \frac{\partial}{\partial \theta^{a}} + i (\sigma^{\mu})_{a}^{\:\: \dot{b}} \bar{\theta}_{\dot{b}}  \frac{\partial}{\partial x^{\mu}} \hspace{15pt}\textnormal{and}\hspace{15pt}\bar{Q}^{\dot{a}} = \frac{\partial}{\partial \bar{\theta}_{\dot{a}}} + i \theta^{b} (\sigma^{\mu})_{b}^{\:\: \dot{a}} \frac{\partial}{\partial x^{\mu}},
\end{equation}

as the vector fields that ``implement" or ``generate" the SUSY transformations viz

\begin{equation}\nonumber
\delta \Phi =\left( \epsilon Q + \bar{\epsilon} \bar{Q}\right)[\Phi] =  \delta x \frac{\partial \Phi}{\partial x} + \delta \theta \frac{\partial \Phi}{\partial \theta} + \delta \bar{\theta} \frac{\partial \Phi}{ \partial \bar{\theta}},
\end{equation}

for any arbitrary superfield $\Phi \in C^{\infty}(\mathbbmss{R}^{4|4})$. It is also easy to show that

\begin{equation}
[Q_{a},\bar{Q}^{\dot{b}}] = 2 i(\sigma^{\mu})_{a}^{\:\: \dot{b}} \frac{\partial}{\partial x^{\mu}},
\end{equation}

and that all other commutators involving $Q, \bar{Q}$ and $\frac{\partial}{\partial x}$ are identically zero. \\

Also fundamental to supersymmetry are  the SUSY covariant derivatives

\begin{equation}
\mathbbmss{D}_{a} =  \frac{\partial}{\partial \theta^{a}} - i (\sigma^{\mu})_{a}^{\:\: \dot{b}}\bar{\theta}_{\dot{b}} \frac{\partial }{\partial x^{\mu}}\hspace{15pt}\textnormal{and}\hspace{15pt}
\bar{\mathbbmss{D}}^{\dot{a}} = \frac{\partial}{\partial \bar{\theta}_{\dot{a}}} - i \theta^{b} (\sigma^{\mu})_{b}^{\:\: \dot{a}} \frac{\partial }{\partial x^{\mu}},
\end{equation}

which can be introduced to compensate for the fact that $\frac{\partial \Phi}{\partial \theta}$ and $\frac{\partial \Phi}{\partial \overline{\theta}}$ do not transform as superfields. This is clear from the transformation rules of the partial derivatives.\\

Up to this point (modulo some conventions) everything has been rather standard and can be found in many textbooks, for example \cite{Wess1992,West1990}. From now on we will employ a more \emph{geometric} point of view in terms of distributions. To the authors knowledge understanding supersymmetry in terms of distributions can be attributed to Manin, at least for supermanifolds of dimension $(1|1)$ and $(1|2)$.

\begin{definition}
The \textbf{SUSY structure} on $\mathbbmss{R}^{4|4}$ is the non-integrable distribution $\D \subset T(\mathbbmss{R}^{4|4})$ of corank $(4|0)$ that can be  spanned by the SUSY covariant derivatives.
\begin{equation}
\D = \Span\left\{ \mathbbmss{D}_{a}, \hspace{5pt}\bar{\mathbbmss{D}}^{\dot{a}}  \right\}.
\end{equation}
\end{definition}

Clearly $\D$ is spanned by four odd vector fields, thus it defines a distribution that is of corank $(4|0)$. That is the SUSY structure has four less even vectors in its local basis as compared to the tangent bundle.  Of course there are many other basis vectors that could be chosen to span the SUSY structure.  The SUSY structure is non-integrable in the sense of  Frobenius as

\begin{equation}\nonumber
[\mathbbmss{D}_{a}, \bar{\mathbbmss{D}}^{\dot{a}}] = -2i (\sigma^{\mu})_{a}^{\:\: \dot{a}}\frac{\partial}{\partial x^{\mu}} \not\subset \D.
\end{equation}

So the SUSY structure is formally very similar to a contact structure on a manifold, apart from the corank. The claim is that one can encode the SUSY structure in a \emph{higher} version of a contact structure.

\section{Supersymmetry and polycontact structures}\label{sec:polycontact}

The standard approach to constructing SUSY covariant derivatives is to employ the methods of homogeneous spaces applied to supermanifolds. One constructs the appropriate Maurer--Cartan form  and then extracts the covariant derivatives. For rigid superspace one can more or less guess the covariant derivatives from their required algebraic properties. The power of coset methods lies in the ability to deal with far more complicated theories, such as supergravity.  An accessible review to the construction of the Maurer--Cartan form associated with $\mathcal{N}=1$ supersymmetric field theory can be found in \cite{Bagger1996}.\\

In \cite{Bruce2011} we employed  the language of contact geometry to investigate  the $d=1$, $N=2$ super-Poincar\'e algebra. In particular it was show that the contact structure on $\mathbbmss{R}^{1|2}$ is contained in the appropriate Maurer--Cartan form  as the component belonging to the stability subgroup generated by temporal translations. Taking this observation as our \emph{cue},  let us quickly examine the relevant  Maurer-Cartan form for the  $d=4$, $\mathcal{N}=1$ super-Poincar\'e group.  \\

Let $G$ be the $d=4$, $\mathcal{N}=1$ super-Poincar\'e group  and let $H$ be the Lorentz group. The group $H$ is a stabiliser subgroup of $G$.  As the supergroup $G$ acts on $\mathbbmss{R}^{4|4}$ transitively we have:\\

\parbox[h][50pt][l]{420pt}{
\emph{The action of the supergroup $G$ can be realised by the left multiplication on the coset $G/H$ and the coordinates which parameterise the coset are given by the coordinates on $\mathbbmss{R}^{4|4}$}.}\\

Then

\begin{equation}
g = e^{i \left( x^{\mu} \frac{\partial}{\partial x^{\mu}} + \theta^{a} Q_{a} + \bar{\theta}_{\dot{a}}\bar{Q}^{\dot{a}}  \right)},
\end{equation}

is a natural parametrisation of the coset $G/H$.\\
\begin{definition}
The  (left invariant) \textbf{Maurer--Cartan form} is the  one form  given by
\begin{equation}\nonumber
i \Omega = g^{-1}\left(dg  \right).
\end{equation}
\end{definition}

We include an overall factor of ``$i$" for convenience.  It is natural to consider the Maurer--Cartan form as a tangent bundle valued one-form on  $\mathbbmss{R}^{4|4}$. To calculate this one appeals to \emph{the Hadamard lemma}, which of course is closely related to the \emph{Baker--Campbell--Hausdorff formula}:

\begin{eqnarray}
\nonumber i \Omega &=& e^{-i \left(  x^{\mu} \frac{\partial}{\partial x^{\mu}} + \theta^{a} Q_{a} + \bar{\theta}_{\dot{a}}\bar{Q}^{\dot{a}}   \right)}\left(  de^{i \left(  x^{\mu} \frac{\partial}{\partial x^{\mu}} + \theta^{a} Q_{a} + \bar{\theta}_{\dot{a}}\bar{Q}^{\dot{a}}   \right)} \right)\\
\nonumber &=& i \left( \left(dx^{\mu} + i \left( \theta^{a}(\sigma^{\mu})_{a}^{\:\: \dot{b}}d\bar{\theta}_{\dot{b}} + d \theta^{a}(\sigma^{\mu})_{a}^{\:\: \dot{b}} \bar{\theta}_{\dot{b}} \right) \right)\frac{\partial}{\partial x^{\mu}}  + d \theta^{a} Q_{a} + d \bar{\theta}_{\dot{a}}\bar{Q}^{\dot{a}}\right).
\end{eqnarray}

Then we ``extract" the part associated with the Lorentz group:

\begin{definition}
The \textbf{SUSY polycontact structure} on $\mathbbmss{R}^{4|4}$ is the vector-valued Grassmann odd one-form
\begin{equation}
\alpha = \left( dx^{\mu} + i \left( \theta^{a}(\sigma^{\mu})_{a}^{\:\: \dot{b}}d\bar{\theta}_{\dot{b}} + d \theta^{a}(\sigma^{\mu})_{a}^{\:\: \dot{b}} \bar{\theta}_{\dot{b}} \right) \right)\frac{\partial}{\partial x^{\mu}}.
\end{equation}
\end{definition}

Note that the SUSY polytcontact structure can be written in the form
\begin{equation}\nonumber
\alpha = d - d\theta^{a}\mathbbmss{D}_{a} - d \bar{\theta}_{\dot{a}}\bar{\mathbbmss{D}}^{\dot{a}},
\end{equation}
where $d$ is the exterior derivative on $\mathbbmss{R}^{4|4}$.  Also notice that setting $\theta = \bar{\theta} = 0$ reduces the SUSY polycontact structure to the exterior derivative on $\mathbbmss{R}^{4}$.\\

The claim is that the SUSY polycontact structure  shares some basic properties with  contact structures. In particular the SUSY polycontact structure   is nowhere-vanishing in the sense that $\alpha|_{\theta = 0, \bar{\theta}=0} \neq 0$. The SUSY polycontact structure can be thought of as being  ``four contact structures" labeled  by a space-time index. Each ``contact form" defines a corank $(1|0)$ hyperplane distribution via its kernel and thus we have in total  a corank $(4|0)$ distribution.

\begin{remark}
The vector-valued two form $d \alpha$ can be considered as a \emph{polysymplectic from} (we briefly discuss such structures in Section(\ref{sec:polysymplectic})) on the space $\mathbbmss{R}^{0|4} \subset \mathbbmss{R}^{4|4}$. That is the exterior derivative of the polycontact form is closed and (as we shall prove) non-degenerate on  the pure odd subspace $\mathbbmss{R}^{0|4}$. This should be compared to the standard contact structure on $\mathbbmss{R}^{3}$ given by $dz + x dy$ in natural coordinates $(x,y,z)$. Note $d (dz + x dy) = dxdy$ is the canonical symplectic structure on $\mathbbmss{R}^{2} \:\: (\subset \mathbbmss{R}^{3})$. This justifies our nomenclature \emph{polycontact}.
\end{remark}

\begin{theorem}\label{theorem1}
\begin{enumerate}
\item The kernel of the SUSY polycontact structure is precisely the SUSY structure.
\begin{equation}\nonumber
\D := \ker(\alpha) = \Span\left\{ \mathbbmss{D}_{a},  \bar{\mathbbmss{D}}^{\dot{a}} \right\} \subset T(\mathbbmss{R}^{4|4}).
\end{equation}\nonumber
\item The exterior derivative of the SUSY polycontact structure is non-degenerate on the SUSY structure in the sense that
\begin{equation}\nonumber
i_{X}\left(d\alpha\right) = 0 \Leftrightarrow X=0,
\end{equation}
if $X \in \D$.
\end{enumerate}
\end{theorem}

\begin{proof}
Consult the Appendix for details about vector-valued differential forms on supermanifolds.
\begin{enumerate}
\item Explicitly
\begin{equation}\nonumber
i_{\mathbbmss{D}_{a}} = - \frac{\partial}{\partial d\theta^{a}} + i (\sigma^{\mu})_{a}^{\:\: \dot{b}}\bar{\theta}_{\dot{b}} \frac{\partial }{\partial dx^{\mu}}.
\end{equation}
Then
\begin{equation}\nonumber
i_{\mathbbmss{D}_{a}}\alpha = \left(i (\sigma^{\mu})_{a}^{\:\: \dot{b}}\bar{\theta}_{\dot{b}} -  i (\sigma^{\mu})_{a}^{\:\: \dot{b}}\bar{\theta}_{\dot{b}}\right) \frac{\partial}{\partial x^{\mu}} = 0.
\end{equation}
An almost identical calculation establishes that $\alpha$ is annihilated by $i_{\bar{\mathbbmss{D}}^{\dot{a}}}$. As it is clear that $\mathbbmss{D}_{a}$ and $\bar{\mathbbmss{D}}^{\dot{a}}$ are linearly independent we establish that they form  a basis for the SUSY structure.
\item It is straightforward to  show that
\begin{equation}\nonumber
d\alpha = 2 i \left(d \theta^{a} (\sigma^{\mu})_{a}^{\:\: \dot{b}}d\bar{\theta}_{\dot{b}} \right)\frac{\partial}{\partial x^{\mu}},
\end{equation}
and thus
\begin{equation}\nonumber
i_{\mathbbmss{D}_{a}}(d\alpha) = -2 i \left( (\sigma^{\mu})_{a}^{\:\: \dot{b}}d\bar{\theta}_{\dot{b}} \right)\frac{\partial}{\partial x^{\mu}}, \hspace{15pt}\textnormal{and}\hspace{15pt} i_{\bar{\mathbbmss{D}}^{\dot{a}}}(d\alpha) = - 2 i \left(d \theta^{b} (\sigma^{\mu})_{b}^{\:\: \dot{a}} \right)\frac{\partial}{\partial x^{\mu}},
\end{equation}
which implies the non-degeneracy condition.
\end{enumerate}
\end{proof}

\begin{proposition}\label{prop1}
The SUSY polycontact on $\mathbbmss{R}^{4|4}$ is invariant under
\begin{enumerate}
\item SUSY transformations:\\
\begin{tabular}{l}
$x'^{\mu} = x^{\mu}+ i \left(\epsilon^{a}(\sigma^{\mu})_{a}^{\:\: \dot{b}} \bar{\theta}_{\dot{b}} - \theta^{a}(\sigma^{\mu})_{a}^{\:\: \dot{b}} \bar{\epsilon}_{\dot{b}}  \right)$ \\
$ \theta'^{a} = \theta^{a}+ \epsilon^{a}$  and $\bar{\theta}'_{\dot{a}} = \bar{\theta}_{\dot{a}}+ \bar{\epsilon}_{\dot{a}}$.
\end{tabular}
\item Poincar$\acute{\textnormal{e}}$ transformations:\\
 \begin{tabular}{l}
 $x'^{\mu} = x^{\nu}\Lambda_{\nu}^{\:\: \mu} + a^{\mu}$.
 \end{tabular}
\item R-transformations: \\ $\theta'^{a} = e^{i \beta}\theta^{a}$, and $\bar{\theta}'_{\dot{a}} = e^{-i \beta} \bar{\theta}_{\dot{a}}$, where $\beta \in \mathbbmss{R}$.
\end{enumerate}
\end{proposition}
\begin{proof}
The above proposition can be proved via direct computation, details are left for the reader.
\end{proof}

We then interpret the SUSY transformations, the  Poincar$\acute{\textnormal{e}}$ transformations and the R-transformations as ``strict polycontactomorphisms", that is they preserve the SUSY polycontact structure, and thus the SUSY structure.

\begin{corollary}
 The vector fields $Q_{a}$ and  $\bar{Q}^{\dot{a}}$ are strict contact vector fields of the the SUSY ploycontact structure, i.e.
 \begin{equation}\nonumber
 L_{Q_{a}}\alpha = 0,   \hspace{15pt}\textnormal{and} \hspace{15pt} L_{\bar{Q}^{\dot{a}}}\alpha=0.
 \end{equation}
\end{corollary}

The above corollary can  also be verified directly. \\

Of course the vector fields $Q$ and $\bar{Q}$  represent ``infinitesimal strict polycontactomorphisms". We think of the SUSY polycontact form as being constant in the ``directions" defined by $Q$ and $\bar{Q}$. This in turn implies that
\begin{center}
\begin{tabular}{ll}
$[Q_{a},\mathbbmss{D}_{b}] = 0$, &  $[\bar{Q}^{\dot{a}}, \bar{\mathbbmss{D}}^{\dot{b}}] =0$,\\
$[Q_{a}, \bar{\mathbbmss{D}}^{\dot{b}}] =0$ & $[\bar{Q}^{\dot{a}}, \mathbbmss{D}_{b}] =0$,
\end{tabular}
\end{center}

as required by covariant derivatives.  This can of course be verified directly using local coordinates if desired.  \\

Recall that the exterior derivative of the  SUSY polycontact structure is non-degenerate on the SUSY structure. This is the analogue of a hyperplane distribution being ``maximally non-integrable" and thus being described in terms of a classical contact structure.\\

From the classical theory of contact structures, we know that there is a privileged strict contact  vector field known as the \textbf{Reeb vector field}. The generalisation to the SUSY polycontact structure is as follows. Instead of a single Reeb vector we have four Reeb vectors indexed by space-time. Let us denote these vector as $P_{\mu}$, the reason for this will become clear. These vector fields are defined by the condition that

\begin{equation}\nonumber
i_{P_{\mu}}\alpha = \frac{\partial}{\partial x^\mu} \hspace{15pt}\textnormal{and} \hspace{15pt} i_{P_{\mu}}(d\alpha) =0.
\end{equation}

\begin{proposition}\label{prop2}
On $\mathbbmss{R}^{4|4}$ equipped with the SUSY polycontact structure  the Reeb vector fields are given by $P_{\mu} = \frac{\partial }{\partial x^{\mu}}$ in natural coordinates.
\end{proposition}

\begin{proof}
Via direct computation.
\end{proof}

Thus we see that the  Reeb vector fields corresponds to space-time translations. We are then led to an interesting interpretation of the  so-called \emph{right super-translation and space-time-translation algebra}:

\begin{eqnarray}
[Q_{a},\bar{Q}^{\dot{b}}] &=& 2 i(\sigma^{\mu})_{a}^{\:\: \dot{b}} P_{\mu},\\
\nonumber [Q_{a}, P_{\mu}] &=& [\bar{Q}^{\dot{a}}, P_{\mu}] = 0,
\end{eqnarray}

as a  Lie subalgebra of the Lie algebra of   polycontact vector fields of the SUSY polycontact structure.\\

\begin{remark}
One can also discuss  R-symmetry infinitesimally as being implemented by the vector field\newline  $R := i \left( \bar{\theta}_{\dot{a}}\frac{\partial}{\partial \bar{\theta}_{\bar{a}}}   -  \theta^{a} \frac{\partial}{\partial \theta^{a}}\right) \in \Vect(\mathbbmss{R}^{4|4})$. Then:

\begin{center}
\begin{tabular}{ll}
$[R,Q_{a}] = i Q_{a}$,   &   $[R,\bar{Q}^{\dot{a}}] = -i \bar{Q}^{\dot{a}}$,\\
$[R,P_{\mu}] =0$ & $[R,R]=0$.
\end{tabular}
\end{center}
In short, $R$-symmetry can also be understood in terms of the Lie algebra of strict polycontact vector fields.\\
\end{remark}

Any vector field  $X \in \Vect(\mathbbmss{R}^{4|4})$ can be written in natural coordinates as

\begin{equation}\nonumber
X = X^{\mu}\frac{\partial}{\partial x^{\mu}} + X^{a}\frac{\partial}{\partial \theta^{a}} + \bar{X}_{\dot{a}} \frac{\partial}{\partial \bar{\theta}_{\dot{a}}}.
\end{equation}

Then using the local expressions for the SUSY covariant derivatives any vector field can be cast in the form
\begin{equation}
X =   \left( X^{\mu} + i \left( X^{a}(\sigma^{\mu})_{a}^{\:\: \dot{b}}\bar{\theta}_{\dot{b}}  +  \bar{X}_{\dot{a}} \theta^{b}(\sigma^{\mu})_{b}^{\:\: \dot{a}} \right) \right)\frac{\partial}{\partial x^{\mu}}  +  X^{a} \mathbbmss{D}_{a}+ \bar{X}_{\dot{a}} \bar{\mathbbmss{D}}^{\dot{a}},
\end{equation}

quite independently of the SUSY polycontact structure. The remarkable point is that this decomposition can be understood in an analogous way to the decomposition of vector fields on a classical contact manifold.   \\

\begin{proposition}\label{prop3}
$\Vect(\mathbbmss{R}^{4|4})$  (i.e. sections of the tangent bundle) decomposes as
\begin{equation}\nonumber
 \Vect(\mathbbmss{R}^{4|4}) = \ker(\alpha)\oplus \ker(d \alpha).
\end{equation}
\end{proposition}

\begin{proof}
Follows from Theorem(3.\ref{theorem1}) and Proposition(3.\ref{prop2}).
\end{proof}

As for examples, it is easy to verify the decompositions

\begin{eqnarray}
\nonumber Q_{a} &=& \mathbbmss{D}_{a} + 2 i \:(\sigma^{\mu})_{a}^{\:\: \dot{b}} \bar{\theta}_{\dot{b}} P_{\mu},\\
\nonumber \bar{Q}^{\dot{a}} &=& \bar{\mathbbmss{D}}^{\dot{a}} + 2 i\: \theta^{b}(\sigma^{\mu})_{b}^{\:\: \dot{a}}P_{\mu},
\end{eqnarray}
and
\begin{equation}\nonumber
R = i \left( \bar{\theta}_{\dot{a}} \bar{\mathbbmss{D}}^{\dot{a}} - \theta^{a}\mathbbmss{D}_{a} \right) + i \left(-2i \: \theta^{a}(\sigma^{\mu})_{a}^{\:\: \dot{b}} \bar{\theta}_{\dot{b}} \right)P_{\mu}.
\end{equation}

\section{Polysymplectization}\label{sec:polysymplectic}

It is well known that extending the dimensions of a contact manifold by $\mathbbmss{R}$ one can canonically construct a symplectic manifold. This procedure is known as \emph{symplectization}.  By minor modification of the standard procedure a polysymplectic structure can be constructed on $\mathbbmss{R}^{4|4}\otimes \mathbbmss{R}$.\\

The definition of a polysymplectic structure that we will employ here is rather na\"{\i}ve.

\begin{definition}
A  vector-valued Grassmann even two-form  $\omega$ on a supermanifold $M$, is said to be a \textbf{polysymplectic form} if and only if
\begin{enumerate}
\item it is closed: $d\omega =0$.
\item it is non-degenerate: $i_{X}\omega = 0 \Longleftrightarrow X =0$ for all $X \in \Vect(M)$.
\end{enumerate}
\end{definition}

\begin{theorem}
Let $M = \mathbbmss{R}^{4|4} \otimes \mathbbmss{R}$ and let $\pi : M \rightarrow \mathbbmss{R}^{4|4}$ be the natural projection. Then the vector-valued even two-form $\omega = d \left( e^{\lambda} \pi^{*}\alpha \right)$ is a polysymplectic form on $M$. Here $\lambda$ is the coordinate on $\mathbbmss{R}$.
\end{theorem}

\begin{proof}
Using the Leibniz rule we have $\omega = e^{\lambda}\left( d \lambda \: \pi^{*}\alpha + \pi^{*}d \alpha \right)$.  From the construction it is clear that the vector-valued two-form in question is Grassmann even and closed. The only question is the non-degeneracy property. First note that we have the decomposition

\begin{equation}\nonumber
TM = \ker(\alpha) \oplus \ker(d \alpha) \oplus \mathbbmss{R},
\end{equation}
which follows from  Proposition(3.\ref{prop3}).  Then consider $d \lambda \: \pi^{*}\alpha + \pi^{*}d \alpha$. Notice that $\pi^{*}\alpha$ is zero on $\ker(\alpha)$, but non-degenerate on $\ker(d \alpha)$. Similarly, $\pi^{*}(d \alpha)$ is zero on $\ker(d \alpha)$, but non-degenerate on $\ker(\alpha)$. The differential $d \lambda$ is non-degenerate on $\mathbbmss{R}$. Furthermore the factor $e^{\lambda}$ is strictly positive and thus the vector-valued two-form $\omega = d \left( e^{\lambda} \pi^{*}\alpha \right)$ is non-degenerate on $M$.
\end{proof}

Explicitly in local coordinates this polysymplectic form is given by
\begin{equation} \nonumber
 \omega = e^{\lambda}\left(  d\lambda \: dx^{\mu} + i d \lambda \left( \theta^{a}(\sigma^{\mu})_{a}^{\:\: \dot{b}}d \bar{\theta}_{\dot{b}} + d \theta^{a} (\sigma^{\mu})_{a}^{\:\: \dot{b}} \bar{\theta}_{\dot{b}} \right)  + 2 i\: d \theta^{a} (\sigma^{\mu})_{a}^{\:\: \dot{b}} d \bar{\theta}_{\dot{b}} \right)\frac{\partial}{\partial x^{\mu}}.
\end{equation}

The above theorem then allows one to associate in a canonical way a polysymplectic structure with $\mathcal{N}=1$ supersymmetry, but at the cost of adding another even (bosonic) degree of freedom to the superspace.

\begin{remark}
One can also consider the natural analogue of the \emph{symplectic cone}. Consider the space $N=\mathbbmss{R}^{4|4}\otimes (0,\infty)$ and the vector-valued two-form $\varpi= d \left( r^{2}\: \alpha \right)$. Here $r$ is the coordinate on the second factor. By the same reasoning as above this form is non-degenerate and thus a polysymplectic form. The multiplicative group $(0, \infty)$ acts on $N$ via $\Phi_{t}: N \rightarrow N$ (for some fixed $t$)  via $\Phi^{*}_{t} r' = t \: r$. It is easy to see that the action of this group on polysymplectic form is $\Phi_{t}^{*}\varpi = t^{2} \varpi$, that is via a dilation. The polysymplectic supermanifold $(N, \varpi)$ is the \emph{polysymplectic cone} associated with the SUSY polycontact structure on $\mathbbmss{R}^{4|4}$.
\end{remark}

\newpage

\section{Concluding remarks}\label{sec:conclusion}
We have shown that  $\mathcal{N}=1$ supersymmetry as formulated in  flat superspace can be, at least at the classical level,   interpreted in terms of a \emph{higher order contact-like structure}. In particular
\begin{enumerate}
\item The SUSY polycontact form is  the piece of the Maurer--Cartan form associated with the
Lorentz subgroup of the super-Poincar\'e group.
\item The SUSY covariant derivatives $\mathbbmss{D}_{a}$  and $\bar{\mathbbmss{D}}^{\dot{a}} \in \Vect(\mathbbmss{R}^{4|4})$ are understood to be a basis for the distribution associated with the SUSY polycontact structure.
    \item The $\mathcal{N}=1$ SUSY algebra is understood in terms of the Lie algebra of ``strict polycontact vector fields" of the SUSY polycontact structure.
\end{enumerate}

It is not obvious what should play the role of Hamiltonian functions in relation with the SUSY polycontact structure. Taking polysymplectic geometry as our cue one would expect there to be some notion of a ``Hamiltonian section", which is of course a vector field itself. However, it is not clear how to formulate this as a  na\"{\i}ve   replacement of a function with a vector field in the definitions associated with contact geometry is nonsensical.\\

The physical relevance of this work is not clear at present. For instance it would be nice to understand if specific actions or even partition functions have some novel geometric  interpretation within the framework of polycontact structures. The study of concrete actions in this context  remains unexplored. Similarly, the physical significance of the associated polysymplectic structure is not at present clear. \\

Generalising to higher or lower dimensional space-times and extended supersymmetries appears straightforward. However, the superspace methods for extended supersymmetry are less developed and become clumsy.  Typically to get at irreducible representations one employs on-shell constraints; we have the problem of auxiliary fields to deal with. Understanding off-shell supersymmetry is more involved and  no complete general theory exists. Indeed the problem of off-shell irreducible representations, that is without direct reference to an action and the related equations of motions, has largely been over looked.\\

We have not in anyway attempted to understand the general theory of polycontact forms on supermanifolds. Polysymplectic geometry was initially developed by G$\ddot{\textnormal{u}}$nther \cite{Gunther1987} to describe a Hamiltonian approach to first order classical field theories. This notion has since been refined to k-symplectic geometry by Awane \cite{Awane1992} and to k-cosymplectic by Le$\acute{\textnormal{o}}$n and collaborators \cite{Leon1998}. There is a closely related notion of multisymplectic geometry, see Gotay, Isenberg \&  Marsden \cite{Gotay2004} for an overview. All these generalisations of classical symplectic geometry have their origin in covariant Hamiltonian formulations of classical field theory. The theory of higher contact-like structures is far less developed, though we must point out the work of van Erp \cite{Erp2011} as this is close to the notions we put forward here. \\

To the authors' knowledge  there has been few  works generalising these structures to supermanifolds. For example, Hrabak developed a multisymplectic approach to the BRST symmetry \cite{Hrabak1999a,Hrabak1999b}, based on Kanatchikov's formalism \cite{Kanatchikov1996}. BV quantisation was studied by Bashkirov \cite{Bashkirov2004}.   Understanding how to formulate supersymmetric field theory or superstring theory in a multisympelctic formalism is missing from the literature.\\

For sure understanding the general theory of vector-valued contact-like structures and their relation to vector-valued symplectic-like structures  will require refinements of the ideas put forward  in this paper.  What ever happens, these initial results  suggests a tantalising and intriguing  link between supersymmetry and \emph{higher contact-like  structures}.  \\

\noindent \textbf{Acknowledgements} \\
The author would like to thank E. van Erp, J.M. Figueroa--O'Farrill and  R.J. Szabo for enlightening email conversations.

\appendix
\section*{Appendix}\label{appendix}
\subsection*{Vector-valued forms on supermanifolds}
In this appendix we present the bare minimum required to understand vector-valued differential forms on supermanifolds as required in the main text. For  a good review of classical vector-valued differential forms and operations upon them  consult Kol\'{a}\v{r}, Michor and Slov\'{a}k \cite{Kolar1993}.\\

Consider a supermanifold $M$, equipped with local coordinates $(x^{A})$, where $\widetilde{x^{A}} = \widetilde{A} \in \mathbbmss{Z}_{2}$. Changes of local coordinates are of the form $\overline{x}^{A} = \overline{x}^{A}(x)$, employing the standard abuses of notation. In essence changes of coordinates on a supermanifold can be written in the same way as coordinate changes on a manifold, the key difference is that one must respect the $\mathbbmss{Z}_{2}$ grading.   We define \emph{differential (pseudo)forms} on a  supermanifold to be functions on the total space of the antitangent bundle $\Pi TM$. Here $``\Pi"$ is the parity reversion functor, it shifts the parity of the fibre coordinates of a vector bundle while leaving the base coordinates unchanged. This definition of differential forms is more than adequate for our purposes. \\

The antitangent bundle is then equipped with natural local coordinates $(x^{A}, dx^{A})$, where the fiber coordinate has the parity $\widetilde{dx^{A}} = \widetilde{A}+1$. The changes of local coordinate induce vector bundle automorphisms of the form

\begin{equation}\nonumber
\overline{x}^{A} = \overline{x}^{A}(x),  \hspace{30pt}\textnormal{and}\hspace{30pt} \overline{dx}^{A} = dx^{B}\left(\frac{\partial \overline{x}^{A}}{\partial x^{B}} \right).
\end{equation}

Pseudoforms need not be polynomial in  the ``differentials" associated with odd coordinates nor is their a notion of a top form. Pseudoforms are not the general objects of integration on supermanifolds.   \\

We employ a local description to define vector-valued differential forms. These in fact could be pseudoforms, but we will not make an issue over this distinction.  In natural local coordinates a vector-valued differential form is given by

\begin{equation}
\Omega = \Omega^{A}(x,dx)  \frac{\partial}{\partial x^{A}}.
\end{equation}

The partial derivative transforms as

\begin{equation}\nonumber
\frac{\partial}{\partial \overline{x}^{A}} = \left(\frac{\partial x^{B}}{\partial \overline{x}^{A}}\right)\frac{\partial}{\partial x^{B}}.
\end{equation}

Importantly, the notion of the exterior derivative,  interior derivative and Lie derivative generalise from acting of differential forms to acting on vector-valued forms. We define the action of the exterior derivative as

\begin{equation}
d \Omega = \left( dx^{B} \frac{\partial \Omega^{A}}{\partial x^{B}}\right)\frac{\partial}{\partial x^{A}}.
\end{equation}

It is straightforward to see that the exterior derivative is a \emph{differential}: $d^{2}=0$. For any vector field $X \in \Vect(M)$ we define the interior derivative as

\begin{equation}
i_{X}\Omega = \left((-1)^{\widetilde{X}} X^{B}\frac{\partial \Omega^{A}}{\partial dx^{B}}\right)\frac{\partial}{\partial x^{A}}.
\end{equation}

The Lie derivative extends via the derivation rule,

\begin{eqnarray}
L_{X}\Omega &=& \left( d(i_{X}\Omega^{A}) - (-1)^{\widetilde{X}+1}i_{X}(d\Omega^{A}) - (-1)^{\widetilde{X}\widetilde{\Omega}}\Omega^{B}\frac{\partial X^{A}}{\partial x^{B}}  \right) \frac{\partial}{\partial x^{A}}\\
\nonumber &=& \left( (-1)^{\widetilde{X}}  dx^{B}\left(  \frac{\partial X^{C}}{\partial x^{B}}\right)\frac{\partial \Omega^{A}}{\partial dx^{C}}+ X^{B}\frac{\partial \Omega^{A}}{\partial x^{B}}  - (-1)^{\widetilde{X}\widetilde{\Omega}}\Omega^{B}\frac{\partial X^{A}}{\partial x^{B}} \right)\frac{\partial}{\partial x^{A}}.
\end{eqnarray}


\begin{thebibliography}{10}
\begin{small}
\bibitem{Arnold1989}
V. I. Arnold.
\newblock{ ``Mathematical Methods of Classical Mechanics", Second Edition,}
\newblock{Springer-Verlag 1989.}
\newblock{Translated by K. Vogtman and A. Weinstein.}

\bibitem{Awane1992}
A. Awane.
\newblock{k-symplectic structures}
\newblock{ \emph{J. Math. Phys.} \textbf{33} (1992) 4046--4052}

\bibitem{Bagger1996}
J. A. Bagger.
\newblock{Weak-Scale Supersymmetry: Theory and Practice,}
\newblock{\texttt{arXiv:hep-ph/9604232v2}, 1996.}

\bibitem{Bashkirov2004}
D. Bashkirov.
\newblock{BV quantization of covariant (polysymplectic) Hamiltonian field theory}.
\newblock{\emph{Int. J. Geom. Meth. Mod. Phys.} \textbf{1} (2004) 233-252}

\bibitem{Bruce2011}
Andrew James Bruce.
\newblock{Contact structures and supersymmetric mechanics}.
\newblock{\texttt{arXiv:1108.5291 [math-ph]}, 2011}.

\bibitem{Erp2011}
Erik van Erp.
\newblock{Contact structures of arbitrary codimension and idempotents in the Heisenberg algebra}.
\newblock{ \texttt{arXiv:1001.5426v2 [math.DG]}, 2011}.

\bibitem{Etnyre2003}
J. Etnyre.
\newblock{Introductory lectures on contact geometry}.
\newblock{\emph{Proc. Sympos. Pure Math}. \textbf{71}, 2003}.



\bibitem{Gotay2004}
M. J. Gotay, J. Isenberg, J. E. Marsden.
\newblock{ Momentum maps and classical relativistic fields. I: Covariant field theory}.
\newblock{ \texttt{arXiv:physics/9801019v2}, 2004}.

\bibitem{Grabowski2011}
J. Grabowski.
\newblock{Graded contact manifolds and principal Courant algebroids}.
\newblock{\texttt{arXiv:1112.0759 [math.DG]}, 2011}.

\bibitem{Gunther1987}
C. G$\ddot{\textnormal{u}}$nther.
\newblock{The polysymplectic Hamiltonian formalism in field theory
and calculus of variations I: The local case}.
\newblock{ \emph{J. Differential Geom.} \textbf{25}
(1987) 23-53.}





\bibitem{Hrabak1999a}
S. P. Hrabak.
\newblock{On a Multisymplectic Formulation of the Classical BRST symmetry for First Order Field Theories Part I: Algebraic Structures}.
\newblock{\texttt{arXiv:math-ph/9901012v1}, 1999}.

\bibitem{Hrabak1999b}
S. P. Hrabak.
\newblock{On a Multisymplectic Formulation of the Classical BRST Symmetry for First Order Field Theories Part II: Geometric Structures}.
\newblock{\texttt{arXiv:math-ph/9901013v1}, 1999}.

\bibitem{Kanatchikov1996}
I. V. Kanatchikov.
\newblock{Novel algebraic structures from the polysymplectic form in field theory}.
\newblock{GROUP21, ``Physical Applications and Mathematical Aspects of Geometry, Groups and Algebras", vol. 2, eds. H.-D. Doebner e.a. (World Sci., Singapore, 1997) p. 894}.
\newblock{Avaliable as 	\texttt{arXiv:hep-th/9612255v2}}.

\bibitem{Kolar1993}
Ivan Kolar, Peter W. Michor and Jan Slovak.
\newblock{ ``Natural Operations in Differential Geometry ", First Edition,}
\newblock{Springer--Verlag, 1993.} \newblock{Also avaliable at \texttt{http://www.emis.de/monographs/KSM/index.html}}

\bibitem{Leon1998}
M. de Le$\acute{\textnormal{o}}$n, E. Merino, J.A. Oubi$\tilde{\textnormal{n}}$a, P. Rodrigues \& M. Salgado.
\newblock{Hamiltonian systems on k-cosymplectic manifolds}
\newblock{ \emph{J. Math. Phys.} \textbf{39}(2) (1998) 876–893.}



\bibitem{Manin1997}
Y. Manin.
\newblock{``Gauge Field Theory and Complex Geometry", Second Edition}
\newblock{Springer, 1997}


\bibitem{Manin1991}
Y. Manin.
\newblock{``Topics in Noncommutative Geometry"}.
\newblock{M. B. Porter Lectures. Princeton University Press, Princeton, NJ, 1991.}

\bibitem{Mehta2011}
Rajan Amit Mehta.
\newblock{Differential graded contact geometry and Jacobi structures}.
\newblock{ \texttt{arXiv:1111.4705v1 [math.SG]}, 2011.}

\bibitem{Schwarz1992}
A. S. Schwarz.
\newblock{Superanalogs of symplectic and contact geometry and their applications to quantum field theory}.
\newblock{	\texttt{arXiv:hep-th/9406120}, 1992.  }



\bibitem{Wess1992}
J. Wess and J. Bagger.
\newblock{ ``Supersymmetry and Supergravity ", Second Edition,}
\newblock{Princeton University Press, 1992.}

\bibitem{West1990}
P. C. West.
\newblock{ ``Introduction to Supersymmetry and Supergravity ", Second Edition,}
\newblock{World Scientific Pub Co Inc, 1990.}


\bibitem{Varadarajan2004}
V. S. Varadarajan.
\newblock{``Supersymmetry for Mathematicians: An Introduction"}
\newblock{Courant Lecture Notes, Vol. 11, 2004}

\bibitem{Voronov1992}
 Th. Th. Voronov.
 \newblock{``Geometric integration theory on supermanifolds", \emph{Sov. Sci. Rev. C. Math.
 Phys.} \textbf{9}.}
 \newblock{Harwood Academic Publications, 1992.}

\end{small}
\end{thebibliography}
\end{document}